\documentclass[%
 preprint,
 amsmath,amssymb,
 aps,
]{revtex4-1}

\usepackage{color}
\usepackage{graphicx}% Include figure files
\usepackage{dcolumn}% Align table columns on decimal point
\usepackage{bm}% bold math
\usepackage{amssymb}
\usepackage{amsmath}
\usepackage{amsthm}
\usepackage{tensor}
\usepackage{geometry}
\usepackage{titlesec}
\usepackage{enumerate}
\usepackage{graphicx}
\usepackage{comment}
\usepackage{float}

\newtheorem{thm}{Theorem}[section]

\theoremstyle{definition}

\newcommand{\bZ}{\mathbb{Z}}

\newcommand{\mv}{\vspace{-0.03cm}}
\titleformat*{\section}{\Large\bfseries}
\titleformat*{\subsection}{\large\bfseries}

\begin{document}

\title{Dynamics of one-dimensional spin models under the line-graph operator}

\author{Marco A. Javarone}
\email{marcojavarone@gmail.com}

\author{Josh A. O'Connor}
\email{josh.oconnor128@gmail.com}

\date{\today}

\begin{abstract}
\noindent We investigate the application of the line-graph operator to one-dimensional spin models with periodic boundary conditions.
The spins (or interactions) in the original spin structure become the interactions (or spins) in the resulting spin structure.
We identify conditions which ensure that each new spin structure is stable, that is, its spin configuration minimises its internal energy.
Then, making a correspondence between spin configurations and binary sequences, we propose a model of information growth and evolution based on the line-graph operator.
Since this operator can generate frustrations in newly formed spin chains, in the proposed model such frustrations are immediately removed.
Also, in some cases, the previously frustrated chains are allowed to recombine into new stable chains.
As a result, we obtain a population of spin chains whose dynamics is studied using Monte Carlo simulations.
Lastly, we discuss potential applications to areas of research such as combinatorics and theoretical biology.
\end{abstract}

% Uncomment for PACS numbers
%\pacs{05.10.-a, 05.10.Ln, 89.75.-k}
%
\maketitle

\section{Introduction}\label{sec:Introduction}
\noindent The line-graph operator~\cite{reddy01} $\lambda$ takes a graph $G(V,E)$ with $V$ vertices and $E$ edges and sends it to a new graph $G'(V',E')$ with $V'$ vertices and $E'$ edges.
For every edge in $G$ there exists a vertex in $G'$ and if two edges in $G$ meet at a vertex then their corresponding vertices in $G'$ are joined by an edge.
Note that the graphs $G$ and $G'$ are isomorphic if $G$ is a cyclic graph.
In this paper we investigate the application of the operator $\lambda$ to one-dimensional spin models with periodic boundary conditions.
Spin models are a core topic in statistical physics and they have been applied to a wide variety of fields such as neural networks, information theory, computational neuroscience and evolutionary game theory among many other areas~\cite{galam01,galam02,loreto01,agliari01,szolnoki02,perc01,perc02,javarone01,javarone02,javarone03}.
A simple spin model with $N$ spins can be described by the following Hamiltonian 
\begin{equation}\label{eq:hamiltonian}
    H = -\frac{1}{N} \sum_{\langle{i,j}\rangle}^{N} J_{ij} \sigma_i \sigma_j 
\end{equation}
\noindent where the spins $\sigma_i$ are each equal to $\pm1$.
Note that $\langle{i,j}\rangle$ denotes a sum over adjacent spins and $J_{ij}$ describes the interaction between any pair of adjacent spins $\sigma_i$ and $\sigma_j$.
Positive interactions $J_{ij}>0$ are defined ferromagnetic, whereas negative interactions $J_{ij}<0$ are defined antiferromagnetic.
Spin models which include both types of interaction are called spin glasses~\cite{sherrington01,parisi01,contucci02,barra01} and they have found remarkable applications in fields such as neural networks and deep learning, among many others~\cite{bengio01,barra02,barra03}.\vspace{0.5cm}

\noindent These models reach equilibrium once their free-energy $F=U-TS$ is minimised, where $U$, $T$ and $S$ are the internal energy, temperature and entropy of the system, respectively.
The temperature plays a key role in deciding if a state of equilibrium can be reached by maximising $S$ or minimising $U$~\cite{mussardo01}.
At high temperatures the entropy is the leading term, whereas the internal energy is most relevant at low temperatures.
From now on, we consider the regime of low temperatures where the internal energy is described by the Hamiltonian in~(\ref{eq:hamiltonian}).
Different spin models can have different energy landscapes.
For example, there are two equilibrium configurations at low temperatures in the ferromagnetic Ising model corresponding to the two states with every spin aligned in the same direction.
On the other hand, spin glasses have a much richer energy landscape with many more minimum energy states.
For instance, such richness can be useful for storing patterns of information~\cite{agliari02}.
Interestingly, in a spin glass there may be frustrations which make it impossible to stabilise at a local minimum of internal energy.\vspace{0.5cm}

\noindent By making a correspondence between positive spins and the symbol $0$, and between negative spins and the symbol $1$, every spin configuration may be mapped to a binary sequence.
Note that this is just a convention and we could have mapped the symbols the other way around.
With this mapping, we are able to assign information content to a spin structure.
We also introduce a correspondence between spin configurations and cyclic graphs which have a vertex 2-colouring.
Therefore, a configuration of $N$ spins corresponds to a periodic binary sequence with period $N$, or to a cyclic graph with $N$ vertices and a vertex 2-colouring.
By modifying the line-graph operator to produce a vertex 2-colouring of $G'$ from a vertex 2-colouring of $G$, we can use these correspondences to observe how periodic binary sequences and spin chains with periodic boundary conditions evolve under the repeated action of $\lambda$.
It turns out that periodic binary sequences whose length is a power of $2$ have particularly interesting dynamics under $\lambda$.
Line-graph dynamics and some of its properties are discussed in Section~\ref{sec:linegraph}.
Then, in Section~\ref{sec:char-graphs}, we encode this dynamics for each length $N$ into the characteristic graphs $\Gamma_n$.
As an application, we propose a model of information growth and evolution based on $\lambda$ and simple mechanisms for handling frustrations in Section~\ref{sec:model}.
In particular, the proposed model studies the population dynamics of spin chains that evolve under $\lambda$ and we observe connections with the characteristic graphs from Section~\ref{sec:char-graphs}.
Remarkably, line-graph evolution becomes an ergodic process when the process begins with an ordered spin chain whose length is a power of 2.
All the results are presented in Section~\ref{sec:results}.
Finally, our observations and future possible directions are discussed in Section~\ref{sec:conclusions}. 

\section{The line-graph operator}\label{sec:linegraph}
\noindent We now describe the action of the line-graph operator~\cite{harary01,beineke01} on one-dimensional spin chains with periodic boundary conditions (other applications of this operator in the context of statistical physics can be found in~\cite{caravelli01,chiu01,caravelli02}).
Since all of these chains have the topology of cyclic graphs, the action of $\lambda$ must preserve the number of spins and interactions.
As a result, two relevant possibilities for the action of the line-graph operator are:
\begin{enumerate}[(i)]
    \item Transforming interactions into spins
    \item Transforming spins into interactions
\end{enumerate}
In the first case, one new spin chain is produced.
In the second case, two new spin chains are produced and they are related by spin permutation, that is, swapping the positive spins for negative spins and vice versa.
The spin chains produced by (i) and (ii) will, in general, not be the same.
These two transformations can be thought of as inverses of each other and, without loss of generality, we label (i) and (ii) by $\lambda$ and $\kappa$, respectively.
At this point, we invoke the correspondence from Section~\ref{sec:Introduction} between spins and the symbols $0$ and $1$ which leads to a correspondence between spin chains and binary sequences.
A periodic binary sequence $a=(\ldots,a_{-1},a_0,a_1,\ldots)$ with period $n$ satisfies $a_i\in\{0,1\}$ and $a_{i}=a_{i+n}$ for all $i\in\bZ$, and we label it by the repeating binary string $a_1\cdots{}a_n$.
We also consider a correspondence between periodic binary sequences and cyclic graphs with vertex 2-colourings.
For the above sequence, the cyclic graph $G$ corresponding to it has $n$ vertices $v_1,\ldots,v_n$ labelled clockwise where the colour of each vertex $v_i$ is $a_i$ for all $i$.\vspace{0.5cm}
Here, rotational symmetry is imposed to ensure that the first vertex label $v_1$ is arbitrary.
In doing so, the symmetry group of this system becomes the cyclic group $C_n$ generated by the rotation $\rho:a_i\mapsto{}a'_i=a_{i+1}$.
%.
Equivalently, two sequences $a$ and $a'$ are rotationally equivalent if they are related by $a'_i=a_{i+s}$ for some integer $s$.
For example, $\rho(0001)=0010$.
Equivalence classes of periodic binary sequences under $C_n$ are called binary necklaces~\cite{mes01}.
In addition, we denote by $O$ the trivial sequence which has all of its digits equal to zero.
\vspace{0.5cm}

\noindent We are now ready to describe in detail how the line-graph can be applied to 2-colouring graphs.
Due to the previous correspondences, all results obtained here can also be applied to spin chains and binary sequences.
The line-graph operator $\lambda$ takes a 2-colouring of $G$ and produces a 2-coloring of $G'$ which is defined in the following way.
Let $v_1,\ldots,v_n$ denote the vertices of $G$ whose colours are $a_1,\ldots,a_n$, respectively, and label the edge which joins $v_i$ and $v_j$ by $e_{ij}$.
The colour $a'_{ij}$ of the vertex $v'_{ij}$ in $G'$ corresponding to the edge $e_{ij}$ in $G$ is defined as
\begin{equation}\label{eq:colour}
    a'_{ij}\equiv{}a_i+a_j\pmod2
\end{equation}
As a result, $a'_{ij}=0$ when $a_i=a_j$ and $a'_{ij}=1$ when $a_i\neq{}a_j$.
Now let $G$ be a cyclic graph with $n$ vertices.
We can label the edge joining $v_i$ and $v_{i+1}$ by $e_i$ for $i=1,\ldots,n-1$, and the edge joining $v_n$ and $v_1$ by $e_n$.
Denote by $v'_i$ the vertex in $G'$ corresponding to the edge $e_i$ in $G$.
The colouring $a'_i$ of the vertex $v'_i$ in $G'$ is given by
\begin{equation}
    a'_i\equiv\begin{cases}a_i+a_{i+1}\pmod2&\hspace{0.5cm}\text{if}\:\:i=1,\ldots,n-1\\a_n+a_1\pmod2&\hspace{0.5cm}\text{if}\:\:i=n\end{cases}
\end{equation}
For example, $\lambda^2(010110010)=\lambda(111010110)=001111011$.
The following result allows us to study the behaviour of periodic spin chains transforming under $\lambda$.
\begin{thm}\label{thm:line-graph}
    Let $\Sigma$ be a periodic binary sequence of length $n$. Then $\lambda^n(\Sigma)=0$ for all $\Sigma$ if and only if $n$ is a power of 2.
\end{thm}
\begin{proof}
    Suppose that $\Sigma=a_1a_2\ldots{}a_n$. The action of $\lambda$ on the periodic sequence $\Sigma$ may be represented using $n$-dimensional vectors over the finite field $\mathbb{F}_2$ as follows:
    \begin{equation}
        \Lambda_n\Sigma=\left(\hspace{0.1cm}\begin{matrix}1\mv&1\mv&0\mv&\cdots\mv&0\mv&0\mv\\0\mv&1\mv&1\mv&\cdots\mv&0\mv&0\mv\\0\mv&0\mv&1\mv&\cdots\mv&0\mv&0\mv\\\vdots\mv&\vdots\mv&\vdots\mv&\ddots\mv&\vdots\mv&\vdots\mv\\0\mv&0\mv&0\mv&\cdots\mv&1\mv&1\mv\\1&0&0&\cdots&0&1\end{matrix}\hspace{0.1cm}\right)\left(\hspace{0.1cm}\begin{matrix}a_1\mv&\null\mv&\null\mv&\null\mv&\null\mv&\null\mv\\a_2\mv&\null\mv&\null\mv&\null\mv&\null\mv&\null\mv\\a_3\mv&\null\mv&\null\mv&\null\mv&\null\mv&\null\mv\\\vdots\mv&\null\mv&\null\mv&\null\mv&\null\mv&\null\mv\\a_{n-1}\mv&\null\mv&\null\mv&\null\mv&\null\mv&\null\mv\\a_n&\null&\null&\null&\null&\null\end{matrix}\hspace{-1.0cm}\right)
    \end{equation}
    where $\Lambda_n$ describes the action of $\lambda$ on binary sequences with period $n$.  
    Now consider
        \begin{equation}
        R_n=\left(\hspace{0.1cm}\begin{matrix}0\mv&1\mv&0\mv&\cdots\mv&0\mv&0\mv\\0\mv&0\mv&1\mv&\cdots\mv&0\mv&0\mv\\0\mv&0\mv&0\mv&\cdots\mv&0\mv&0\mv\\\vdots\mv&\vdots\mv&\vdots\mv&\ddots\mv&\vdots\mv\\0\mv&0\mv&0\mv&\cdots\mv&0\mv&1\mv\\1&0&0&\cdots&0&0\end{matrix}\hspace{0.1cm}\right)
    \end{equation}
    which is the $n \times n$ rotation matrix representing $\rho$. 
    This implies that $\Lambda_n=I_n+R_n$ where $I_n$ is the $n \times n$ identity matrix and we use ${R_n}^n=I_n$ to obtain
    \begin{equation}
        {\Lambda_n}^n=(I_n+R_n)^n=2I_n+\sum_{i=1}^{n-1}\binom{n}{i}{R_n}^i
    \end{equation}
    The first term is even and the remaining terms are also even if and only if $n$ is a power of 2 as a consequence of Lucas' theorem~\cite{lucas}.
\end{proof}
\noindent Many sequences of period $n$ will terminate at $O$ after fewer than $n$ iterations of $\lambda$. For example, the sequence $1001$ is 4 digits long but it will terminate after just 3 iterations since $\lambda^3(1001)=\lambda^2(1010)=\lambda(1111)=0000$.
To make contact with the notation for spin chains, we can replace $0$ and $1$ with $+1$ and $-1$, respectively, which allows us to write the spin chain $\Sigma=1001$ above as $\Sigma=(-1,+1,+1,-1)$.
Applying $\lambda$ to this three times, we find $\lambda^3(-1,+1,+1,-1)=(+1,+1,+1,+1)$, as expected.\vspace{0.5cm}

\noindent Equation~(\ref{eq:colour}) tells us how to colour the vertices of $G'$.  It is also a manifestation of
\begin{equation}\label{eq:colour_spin}
    J'_{ij} = \sigma_i \cdot \sigma_j
\end{equation}
\noindent where $J'_{ij}$ are the interactions between spins $\sigma'_i$ and $\sigma'_j$ of the new spin structure.
Note that the final stable configuration is the trivial sequence $O=\ldots000\ldots$ using equation~(\ref{eq:colour}), but this is equivalent to $O=(\ldots,+1,+1,+1,\ldots)$ using equation~(\ref{eq:colour_spin}).

\section{Characteristic graphs and equivalence classes}\label{sec:char-graphs}
\noindent In this section, we encode the dynamics of periodic binary sequences of period $n$ under $\lambda$ into characteristic graphs which we label $\Gamma_n$.
For instance, when $n=4$, the equivalence classes under the symmetry group $C_n$ are $\{0000\}$, $\{1111\}$, $\{0101,1010\}$, $\{0011,0110,1100,1001\}$, $\{0001,0010,0100,1000\}$ and $\{0111,1110,1101,1011\}$.
Then, we use representatives from each class to label them: $[0000]$, $[0101]$, and so on.
Moreover, in general, if any two sequences $a_1\cdots{}a_n$ and $b_1\cdots{}b_n$ belong to the same equivalence class, their images under $\lambda$ also belong to the same equivalence class.
For example, $0100$ and $0010$ both belong to $[0001]$ and we see that $\lambda(0100)=1100$ and $\lambda(0010)=0110$ both belong to $[0011]$.
The line-graph dynamics for $n=4$ is shown in Figure~\ref{fig:figure_1}.

\begin{figure}[H]
	\begin{center}
		\includegraphics[width=0.6\textwidth]{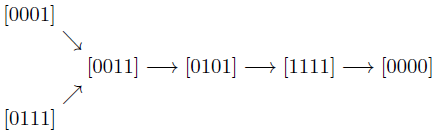}
		\caption{Line-graph dynamics for $n=4$.\label{fig:figure_1}}\vspace{-0.7cm}
	\end{center}
\end{figure}

\noindent In this case, the graph is cycle-free and, as a result, every periodic binary sequence of period $4$ terminates at some `eigensequence' $x$ of $\lambda$ which is defined as a sequence such that $\lambda([x])=[x]$.
Furthermore, this graph is connected so there is exactly one eigensequence for $n=4$.
Other characteristic graphs, for $\lambda$ are shown in Figure~\ref{fig:figure_2}, where we observe the emergence of disconnected components.

\begin{figure}[H]
	\begin{center}
		\includegraphics[width=0.85\textwidth]{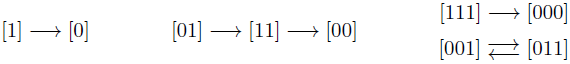}\vspace{-0.3cm}
		\caption{Line-graph dynamics for $n=1$, $n=2$ and $n=3$.\label{fig:figure_2}}\vspace{-0.7cm}
	\end{center}
\end{figure}

\noindent For the sake of simplicity, these digraphs can be reduced to the characteristic graphs $\Gamma_n$ where each vertex is an equivalence class.  The first eight are shown in Figure~\ref{fig:figure_3}.

\begin{figure}[H]
	\begin{center}
		\includegraphics[width=0.9\textwidth]{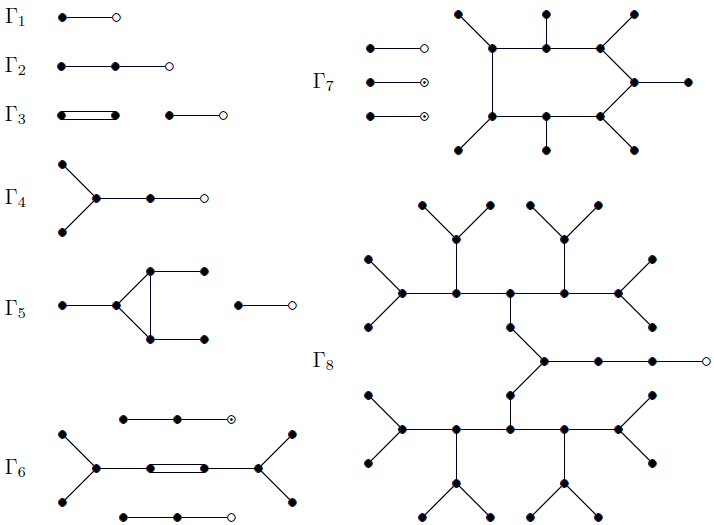}\vspace{-0.2cm}
		\caption{Characteristic graphs for $n=1,\ldots,8$ with $C_n$ symmetry.\label{fig:figure_3}}\vspace{-0.7cm}
	\end{center}
\end{figure}

\noindent In the previous figure, white vertices represent the trivial sequence and white vertices with a dot in the middle represent non-trivial eigensequences of $\lambda$ such as $011011$.\vspace{0.5cm}

\noindent It is now useful to introduce spin permutation symmetry which swaps the symbols 0 and 1 in a binary sequence.
Equivalently, it swaps positive and negative spins in a spin chain.
This symmetry is characterised by the symmetric group $S_2$ generated by the operator $\pi$.
For example, $\pi(0111001)=1000110$.
If we group sequences into equivalence classes under both rotational and spin permutation symmetry, then the symmetry group is $C_n\times{}S_2$ and we find a different collection of characteristic graphs which are shown in Figure~\ref{fig:figure_4}.

\begin{figure}[H]
    \begin{center}
      \includegraphics[width=4.75in]{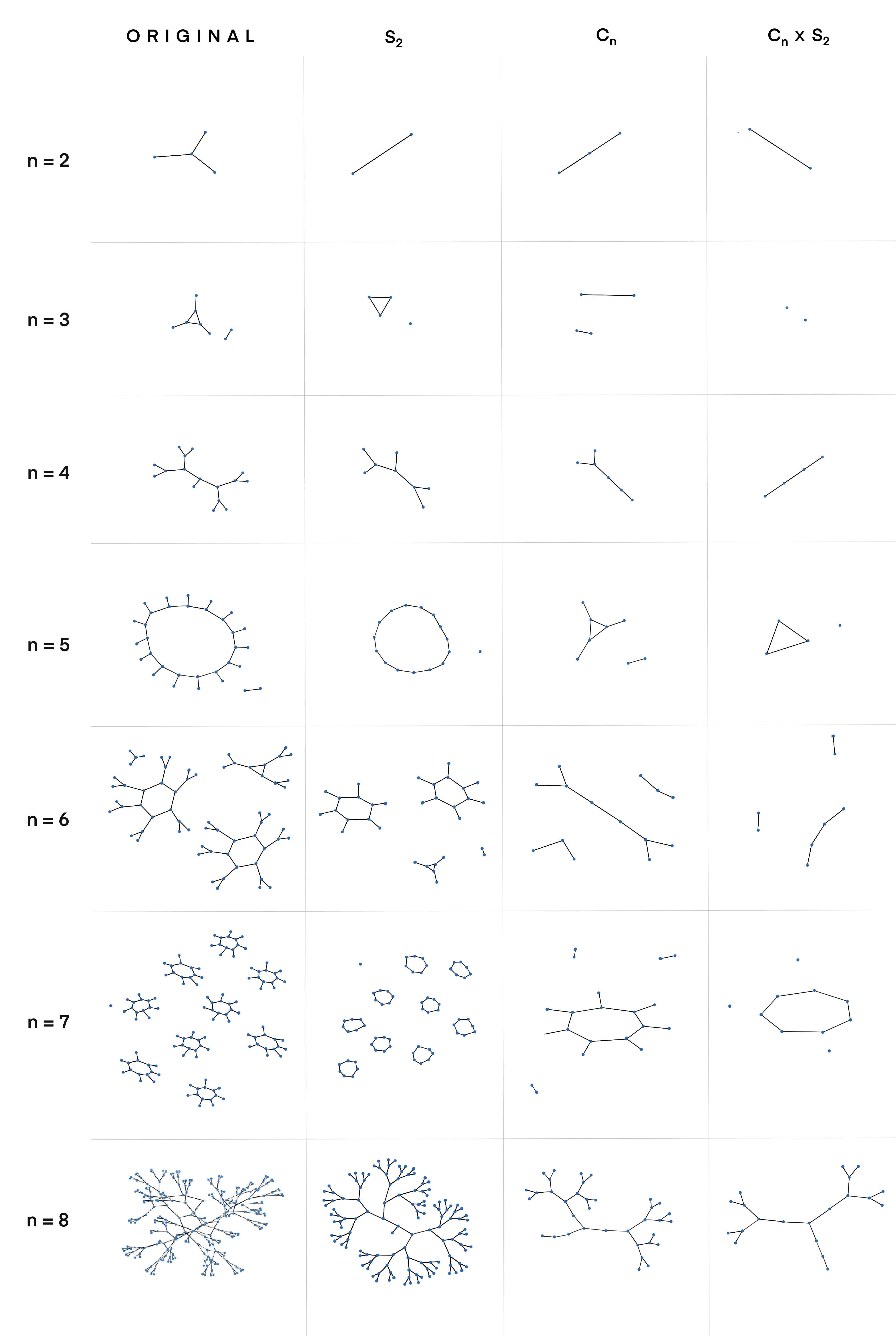}
        \caption{\small From left to right: line-graph dynamics with no symmetry (original), $S_2$ symmetry, $C_n$ symmetry and $C_n\times{}S_2$ symmetry. From top to bottom: sequences with $n=2,\ldots,8$.\label{fig:figure_4}}
    \end{center}
\end{figure}
\noindent A more detailed discussion on this topic can be found in~\cite{gilbert-riordan}.
To conclude the section, we remark that other symmetry groups might be identified.
For example, the dihedral group $D_n$ generated by $\rho$ and the reflection operator $\tau:a_i\mapsto{}a_{n-i}\equiv{}a_{-i}$ whose action is $\tau(a_1\cdots{}a_n)=a_n\cdots{}a_1$.

\section{Population dynamics of spin chains}\label{sec:model}
\noindent In this section we study the population dynamics~\cite{szolnoki01} of spin chains under repeated application of the operator $\kappa$.
Recall that $\kappa$ generates two new spin chains so if we apply $\kappa$ to the spin chain $\Sigma=(+1,-1,-1,+1)$, then we obtain two new spin chains given by $\Sigma_1=(+1,+1,-1,+1)$ and $\Sigma_2=(-1,-1,+1,-1)$.
To verify this, note that $\lambda(\Sigma_1)=\lambda(\Sigma_2)=\Sigma$.
A particularly useful quantity for studying the spin chains that compose our population is the average magnetisation
\begin{equation}
    \langle M \rangle=\frac{1}{N}\sum_{i=1}^{N}\sigma_i
\end{equation}
\noindent which allows us to assess whether or not a configuration is ordered.
In particular, ordered states have $\langle{}M\rangle=\pm1$ and disordered states have $\langle M \rangle \sim 0$.
The sequence of spin chains obtained by repeatedly applying $\kappa$ can be studied by computing the magnetisation at each iteration as shown in Figure~\ref{fig:figure_5}.

\begin{figure}[H]
	\centering
		\includegraphics[width=1.0\textwidth]{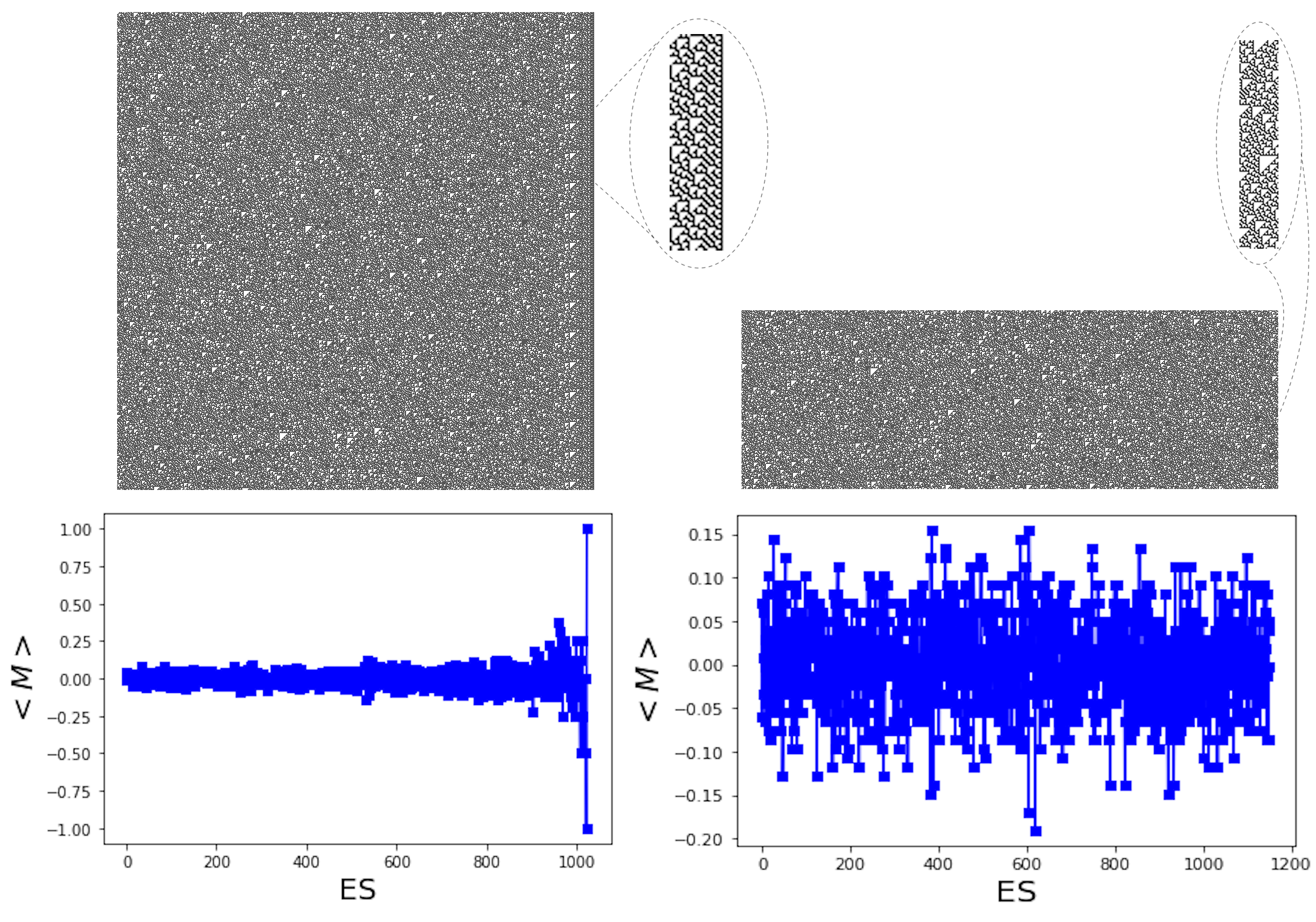}\vspace{-0.2cm}
		\caption{On the left: a binary image (white pixels represent $0$ and black pixels represent $1$) describing the line-graph dynamics of a randomly chosen spin chain with length $N=2^{10}$, and the magnetisation at each evolutionary step is shown below.
		On the right: a similar image for a randomly chosen spin chain with length $N=383$ along with its associated magnetisation.
		The process for $N=383$ was terminated after $3N$ iterations.
		Each inset is an enlargement of the right edge of the related binary image and we observe that the spin chain approaches an ordered state for $N=2^{10}$ but not for $N=383$. \label{fig:figure_5}}
	\vspace{-0.4cm}
\end{figure}

\noindent In this figure, we see the evolution of spin chains of length $N=2^{10}$ and $N=383$.
The length of the first chain is a power of 2 and it terminates after at most $N$ iterations, as expected from Theorem~\ref{thm:line-graph}.
However, for the same reason, the second chain does not terminate into a final stable configuration, so we stop the process after $3N$ iterations.
The magnetisation fluctuates close to zero during the entire simulation for the second chain, but it sharply converges towards $+1$ near the end for the first chain.
This interesting behaviour, expected from Theorem~\ref{thm:line-graph}, can be observed in the two black-and-white images in Figure~\ref{fig:figure_5}.
For instance, notice that the right-most vector in the first image is completely white, but this convergence does not appear in the second image.
It is also worth highlighting that this sort of order-disorder phase transition has an underlying dynamics which is very different from those occurring in other classical spin models (e.g.~\cite{javarone01,javarone02}).
Usually, either spins or interactions are kept \textit{quenched}, while in this model the line-graph operator varies spins and interactions with each iteration.\vspace{0.5cm}

\noindent Now we can proceed to describe the population dynamics of spin chains under repeated action of $\kappa$.
This model may represent a system of information growth and evolution.
In particular, information is encoded into spin chains which evolve using $\kappa$.
The mechanism to form a population is simple: we begin with an arbitrarily chosen spin chain (for simplicity, the initial spin chain could have only positive spins or only negative spins).
Then we apply $\kappa$ to each new spin chain and we only allow new configurations into the population if they are not duplicates.
Moreover, only stable configurations (i.e. those without frustrations) are allowed to generate an offspring.
If we start with a spin chain of length $N$, then the upper bound for the size of the population is $N_{\mathrm{max}}=2^N$
However, from Sections~\ref{sec:linegraph} and~\ref{sec:char-graphs} we know that the only sequences which can reach this upper bound are those whose length is a power of 2.
Once all accessible configurations are generated by $\kappa$ for a given $N$, frustrated configurations will appear in the system.
However, as previously mentioned, the line-graph operator is not applied to them.
To deal with this limitation we introduce the following mechanism: frustrated configurations are broken by removing the interaction responsible for the frustration.
As a result, the topology of each frustrated spin chain becomes linear.
They are now unable to generate offspring since this would not preserve the number of spins and the number of interactions under $\kappa$.
In some cases, these broken chains undergo a process where pairs of them recombine into new stable chains which continue to evolve.
A pair of broken chains will always recombine when the sum of their lengths is a power of 2.
On the other hand, every other pair of broken chains will recombine with a probability denoted $p_r$.
We have summarised the evolution and recombination processes in Figure~\ref{fig:figure_6}.
\begin{figure}[H]
	\centering
		\includegraphics[width=0.7\textwidth]{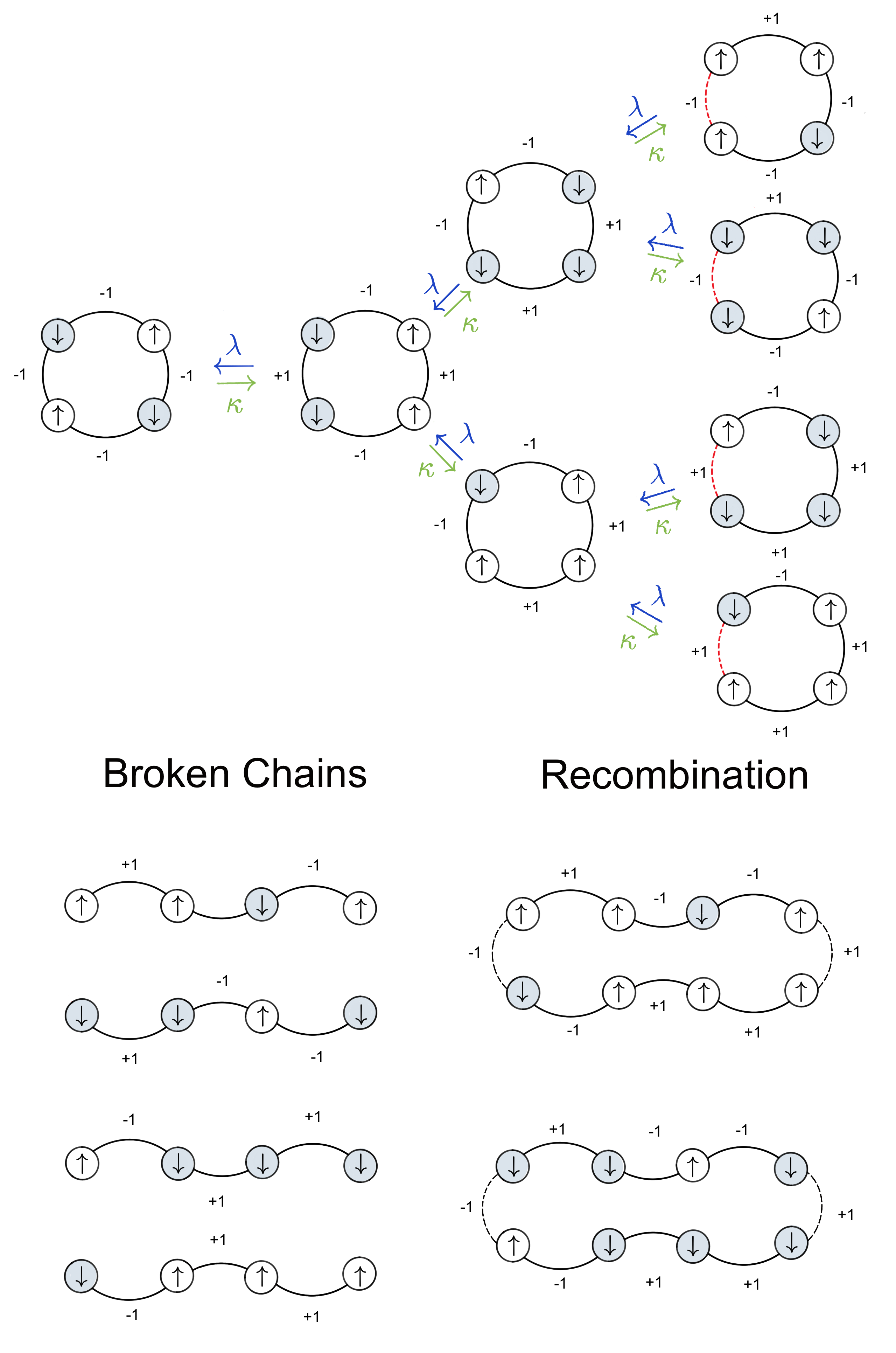}\vspace{-0.2cm}
		\caption{This is a pictorial description of line-graph evolution and recombination.
		The red dotted line represents the interaction to be removed and the black dotted line represents the interaction which is added during recombination.
		Positive spins and negative spins are contained inside white and black vertices, respectively.
		The operators $\lambda$ and $\kappa$ are shown in blue and green, respectively. \label{fig:figure_6}}
\end{figure}

\section{Results}\label{sec:results}
\noindent We implement Monte Carlo simulations to study the proposed model.
Each simulation begins with a single periodic spin chain whose length $N$ ranges from $1$ to $8$ and this is allowed to run for $21$ evolution steps (or time steps).
We compute three relevant quantities while averaging over $1000$ attempts: the number of configurations, the number of frustrated configurations, and the average length of spin chains.
Note that the need to implement Monte Carlo simulations is a result of the stochastic behaviour induced by the recombination mechanism associated to a probability $p_r$.\vspace{0.5cm}
Figure~\ref{fig:figure_7} shows how the number of configurations changes during this evolution for different values of $N$ while $p_r=0$ is fixed.
This value of $p_r$ allows recombination to occur only when the length of the chains is a power of 2.
Recall that this process is ergodic when $N$ is a power of 2 and, as a result, spin chains with length $N\in\{1,2,4,8\}$ are able to generate all possible configurations of that size before the recombination phase which allows them to grow.

\begin{figure}[H]
    \centering
    \includegraphics[width=5.5in]{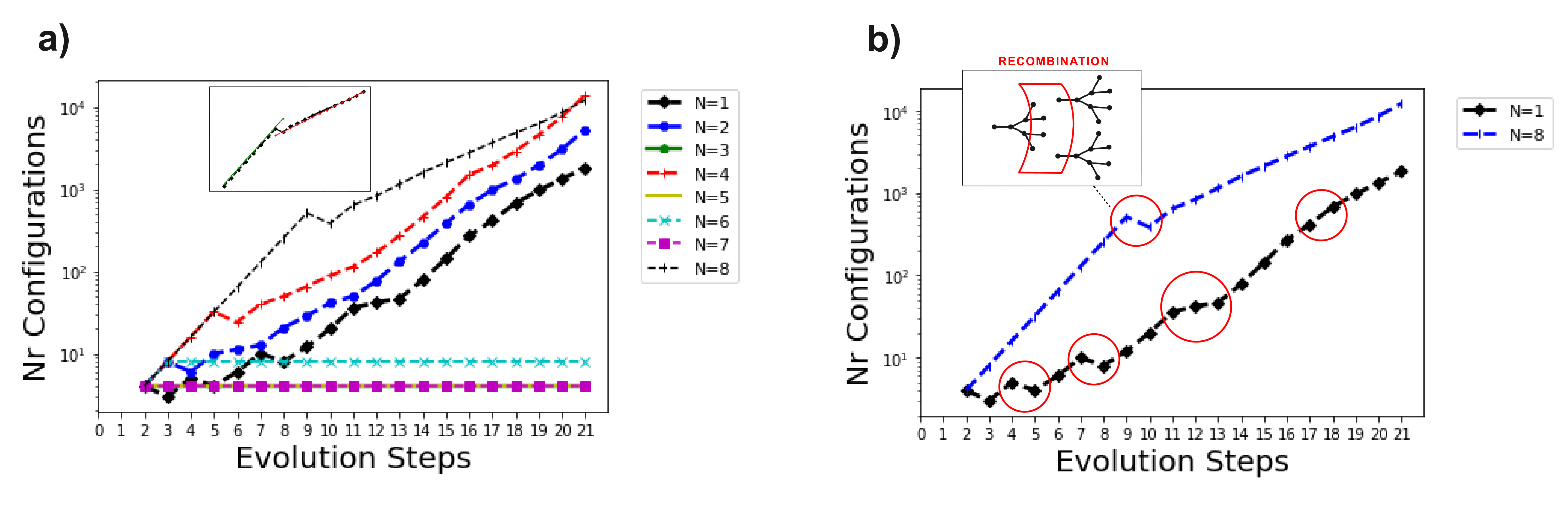}
    \caption{Number of configurations from an initial spin chain with length $N\in\{1,\ldots,8\}$.
    The inset in plot \textbf{a)} shows two curves fitting the curve for $N=8$.
    The red circles in plot \textbf{b)} highlight where the recombination phases occur for $N=1$ and $N=8$.}
    \label{fig:figure_7}
\end{figure}

\noindent In Figure~\ref{fig:figure_7}, the inset in plot \textbf{(a)} highlights the effect of recombination for $N=8$ by showing two functions which fit the curve before and after the recombination phase.
After recombination, the growth in the number of configurations slows down up to approximately $50\%$ and this effect is weaker for smaller values of $N$.
Plot \textbf{(b)} in the same figure emphasises the recombination phases for $N=1$ and $N=8$.
Unless otherwise specified, $N$ refers to the length of the initial spin chain in each simulation.
We now consider the cases where $N\in\{3,5,6,7\}$ by relaxing the constraint on $p_r$.
The result obtained for $p_r\in\{0.0,0.1,\ldots,1.0\}$ are shown in Figure~\ref{fig:figure_8}.

\begin{figure}[H]
    \centering
    \includegraphics[width=5.5in]{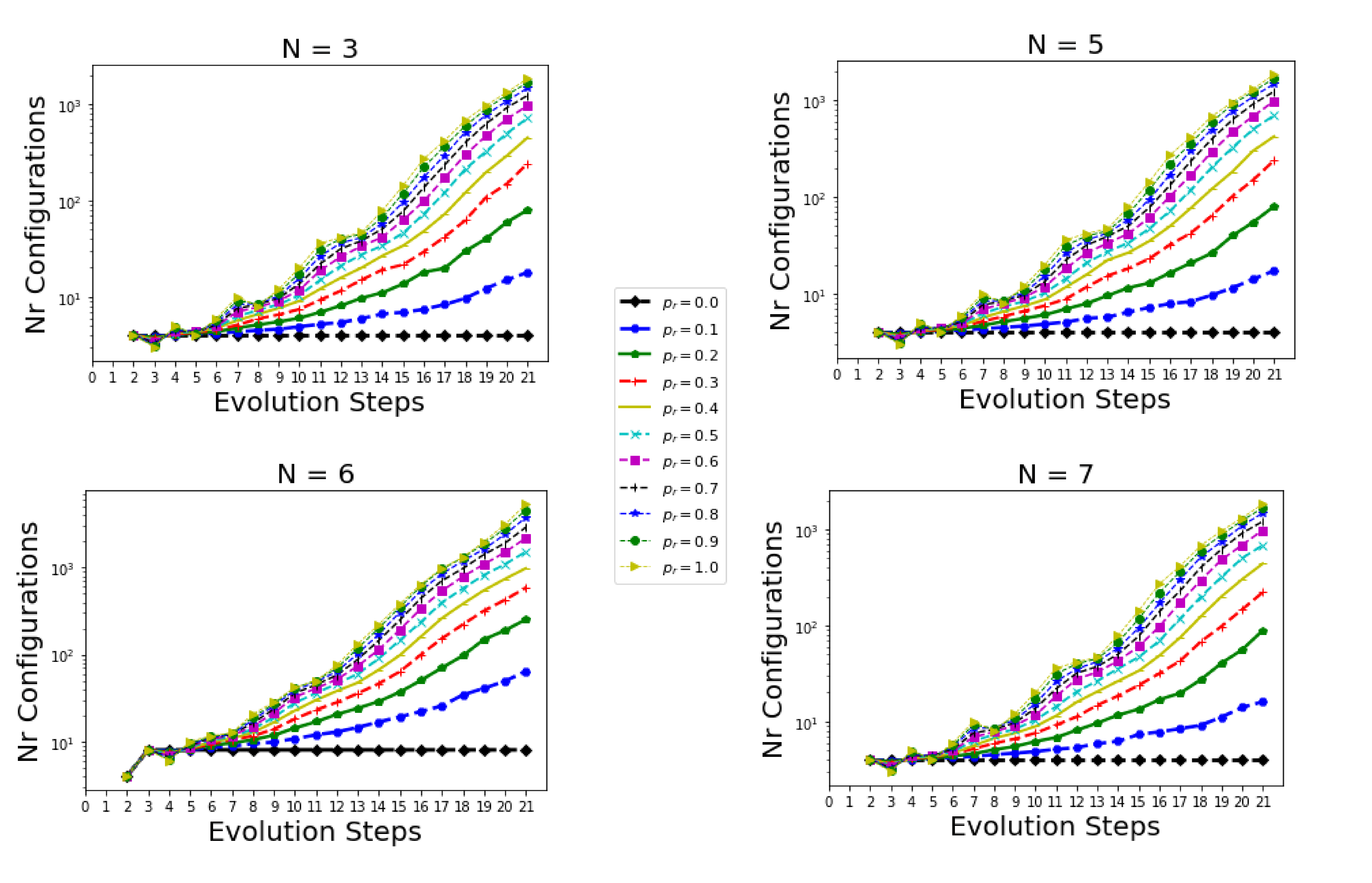}
    \caption{\small Number of configurations emerging over $21$ evolution steps from initial populations composed a single chain of length $N\in\{3,5,6,7\}$.
    As indicated in the legend, there is a curve for each value of $p_r$.\label{fig:figure_8}}
\end{figure}

\noindent In this figure, we observe that a higher value of $p_r$ leads to a larger population, as expected.
In light of this, in Figure~\ref{fig:figure_9}, we compare at different evolution steps the number of configurations obtained for different $N$ while varying $p_r$ from $0$ to $1$.

\begin{figure}[H]
    \centering
    \includegraphics[width=5.5in]{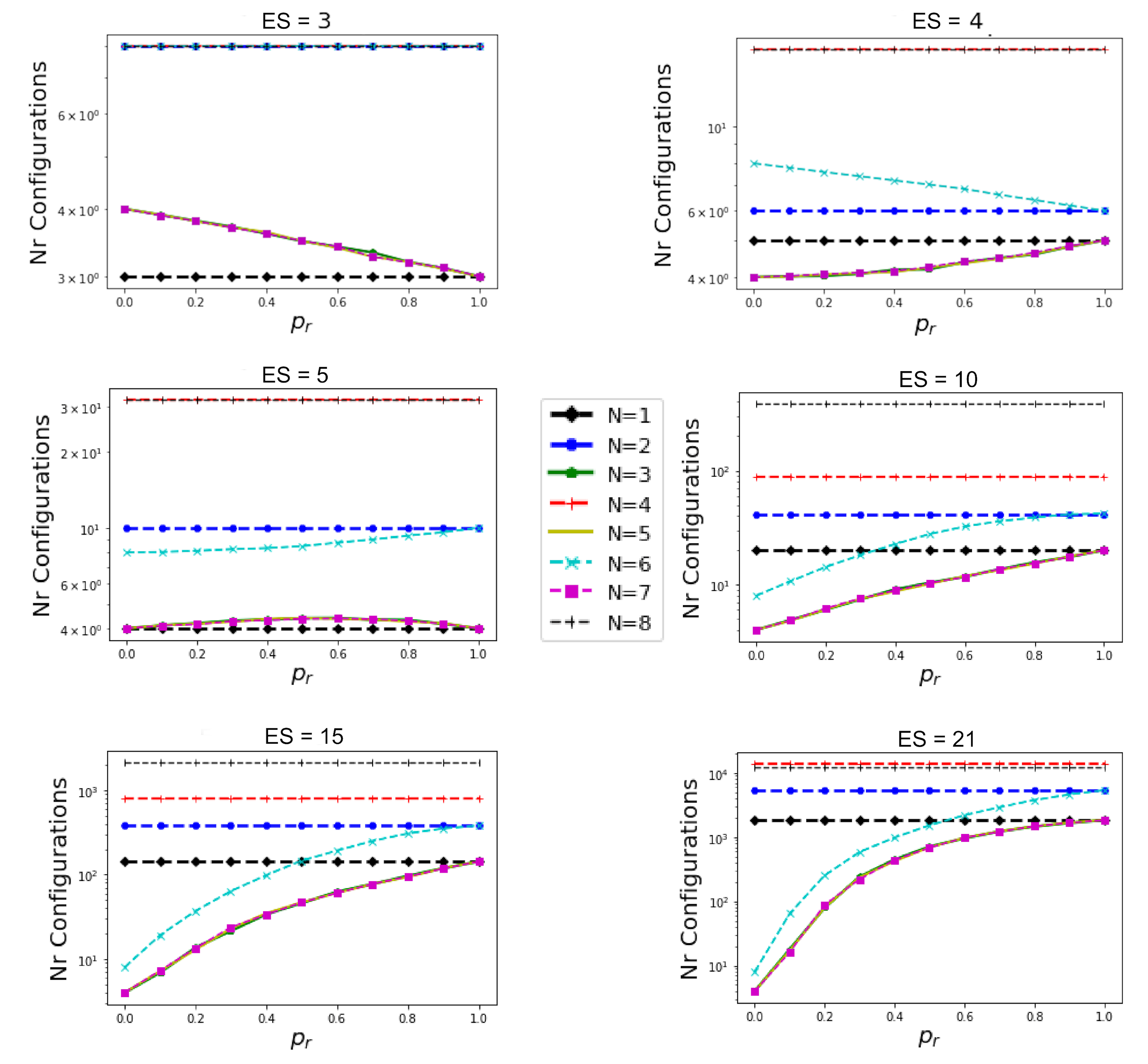}
    \caption{Number of configurations for different values of $p_r$.
    Each plot refers to a specific evolution step $ES$ and each line refers to a different value of $N$.\label{fig:figure_9}}
\end{figure}

\noindent In Figure~\ref{fig:figure_10}, we consider the average spin chain length $\langle{}N(\mathrm{ES})\rangle$ for each value of $N$.
We vary $p_r$ for the first four plots where $N\in\{3,5,6,7\}$.
The last plot considers all values of $N$ from $1$ to $8$ while keeping $p_r=0$.

\begin{figure}[H]
    \centering
    \includegraphics[width=5.5in]{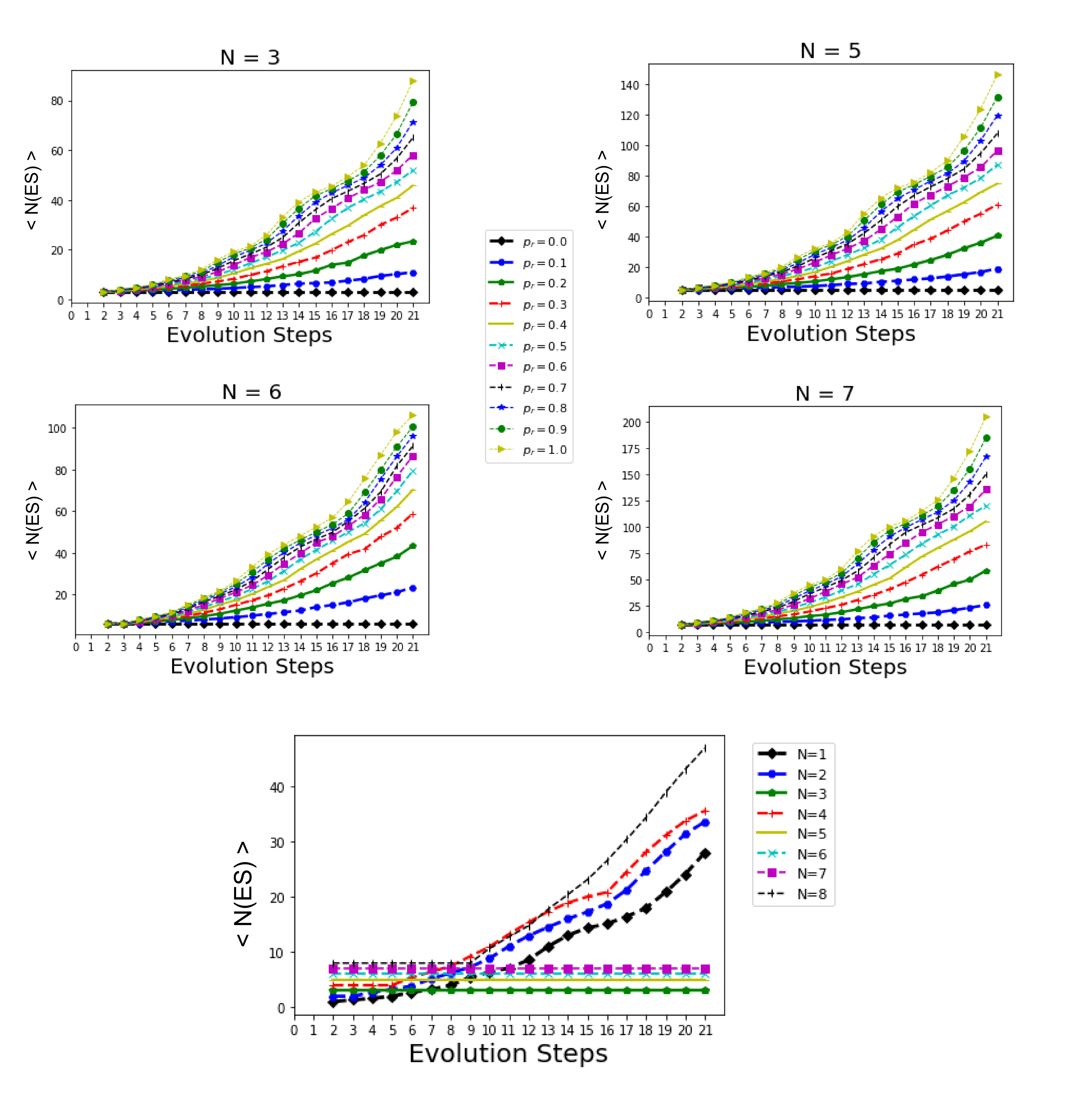}
    \caption{Average length $\langle{}N(ES)\rangle$ of spin chains during line-graph evolution.
    The first four plots refer to $N\in\{3,5,6,7\}$ as $p_r$ varies, and the final plot refers to $N\in\{1,\ldots,8\}$ with $p_r=0$ fixed.\label{fig:figure_10}}
\end{figure}

\noindent Results from $N=1$ to $N=8$ at different evolution steps are shown in Figure~\ref{fig:figure_11}.

\begin{figure}[H]
    \centering
    \includegraphics[width=5.5in]{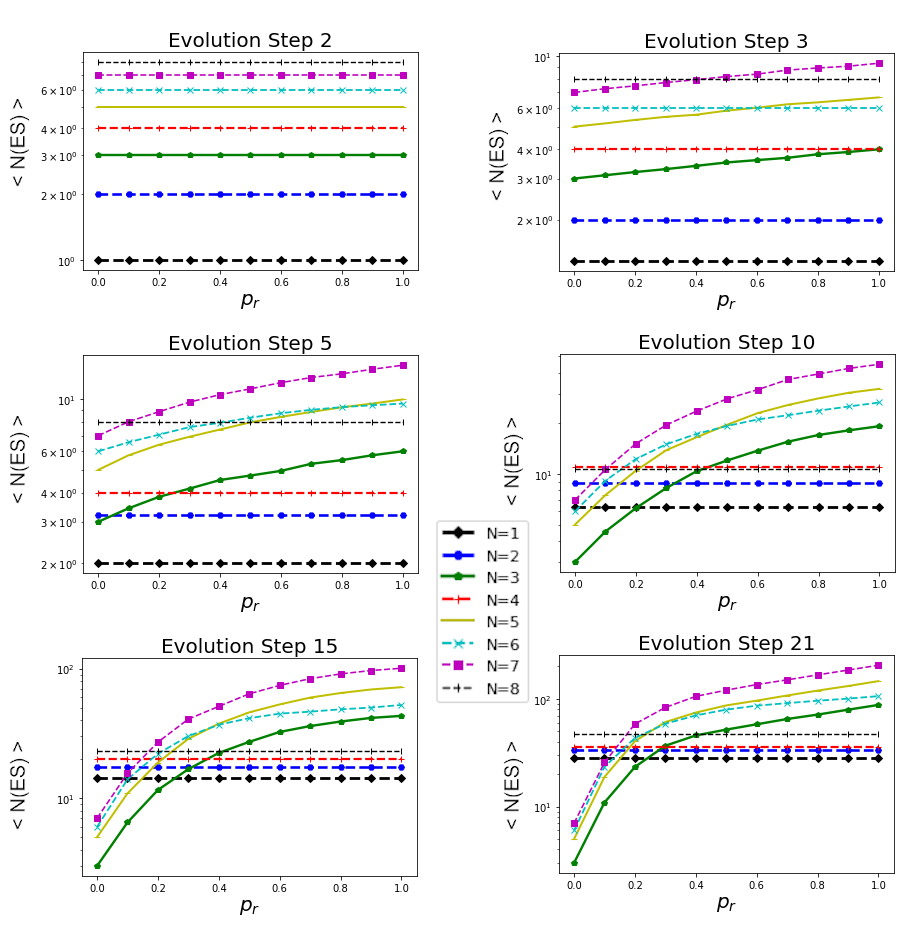}
    \caption{Average length $\langle{}N(ES)\rangle$ of spin chains at different steps during line-graph evolution while $p_r$ varies.
    As indicated in the legend, each line refers to the length $N$ of each initial chain.\label{fig:figure_11}}
\end{figure}

\noindent Figure~\ref{fig:figure_11} shows an interesting phenomenon: when $N$ is not a power of $2$, the average length of spin chains in each population is greater than for populations where $N$ is a power of $2$.
This phenomenon is related to a sort of probability threshold $\hat{p_r}$ which can be observed in the various plots of Figure~\ref{fig:figure_11}.
For instance, when $ES=5$, the thresholds for $N=6$ and $N=7$ are $\hat{p_r}=0.4$ and $\hat{p_r}=0.1$, respectively, and when $ES=21$, the threshold for $N=5$ and $N=6$ are both $\hat{p_r}=0.2$.
Lastly, in Figure~\ref{fig:figure_12}, we consider the number of frustrated configurations emerging in each population.
\begin{figure}[H]
    \centering
    \includegraphics[width=5.5in]{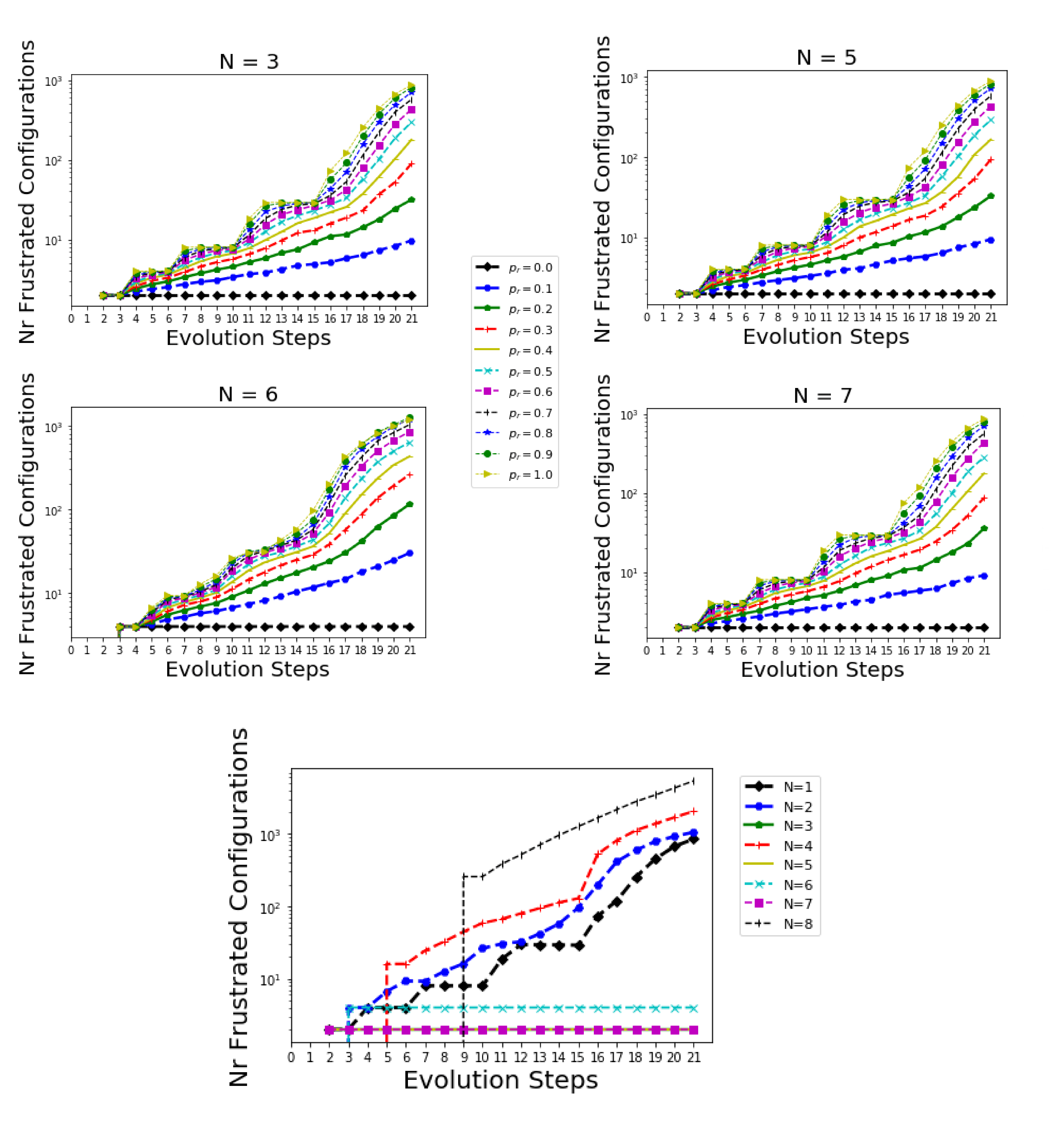}
    \caption{Number of frustrated configurations.
    The first four plots refer to $N\in\{3,5,6,7\}$ as $p_r$ varies, and the final plot refers to $N\in\{1,\ldots,8\}$ with $p_r=0$ fixed.\label{fig:figure_12}}
\end{figure}
\noindent The effect of recombination emerges as a sequence of steps in each curve in some of the above figures, for example Figure~\ref{fig:figure_7} and Figure~\ref{fig:figure_12}.

\section{Conclusion}\label{sec:conclusions}
\noindent This work investigates the action of the line-graph operator $\lambda$ applied to one-dimensional spin models and also studies some of the properties of this operator.
We encode line-graph dynamics into characteristic graphs $\Gamma_n$ based on different symmetry groups, for example $C_n$.
In addition, we propose a model for investigating the dynamics of populations of spin chains which can be thought of as a system of information growth and evolution.
The properties of $\lambda$ are studied analytically, while the evolutionary model is analysed with numerical simulations.
Interestingly, some of the results from our simulations reflect the previously studied properties of $\lambda$.
This is an example of a common general relationship between the topology of a system and dynamical processes within it.
For example, this relationship is well known in complex network theory~\cite{boccaletti01,estrada01,latora01}.
The analysis of $\kappa$ showed that, at some point, frustrated configurations will be generated.
In order to deal with this, the proposed model includes a mechanism for removing frustrations which sometimes allows broken spin chains to recombine into new stable structures.
From this, we are able to observe a system whose information, encoded into spin chains, continuously grow and evolve.
This dynamics resembles that of an automata~\cite{droz01}, as suggested by Figure~\ref{fig:figure_5}, whose properties could be further explored.
We also highlight that $\lambda$ induces a sort of phase transition from disordered configurations to ordered configurations only when the length $N$ of each spin chain is a power of $2$.
This model may also be of interest to those working at the boundary between information theory and physics~\cite{wolpert01,flack01,landauer01,zurek01} as well as for those working in theoretical biology~\cite{douglas01,walker01,walker02,walker03,sole01,niegel01}.
In particular, this model contains spin chains whose evolution alternates between two processes which resemble asexual and sexual reproduction.
The line-graph operator can be seen as an asexual process, whereas the recombination mechanism resembles a sexual process.
Remarkably, some evidence found in nature~\cite{charlesworth01} seems to confirm previous ideas about alternating sexual-asexual reproduction~\cite{green01,neiman01}
The probability of recombination $p_r$ can be seen as a sort of Darwinian fitness which measures the chance for each initial spin chain to have a rich population of descendants.
In this context, we could identify a genealogy tree for every spin configuration under the line-graph operator.
We envision a potential application where patterns of information can be analysed with the aim of computing their `most recent common ancestor'.
For example, this could be used to compare signals or images.\vspace{0.5cm}

\noindent This investigation considers line-graph dynamics for spin chains with the following conditions: they are one-dimensional, they have periodic boundary conditions, and there are two possible spins: $\uparrow$ and $\downarrow$.
We deem that these aspects deserve further attention and, therefore, we plan to investigate line-graph dynamics for spin chains in higher dimensions, with different topologies (see also~\cite{chiu02}), and with more than two spins.
To conclude, we hope that our results and observations stimulate novel ideas in a variety of areas, such as combinatorics, theoretical biology, complex systems and other cross-disciplinary domains.

\section*{Acknowledgments}
\noindent The authors are grateful to Francesco Caravelli for his useful comments and observations.

\end{document}